\def\diag{\mathrm{diag}\mathop}

\documentclass[tightenlines,prl,twocolumn,%
notitlepage]{revtex4-1}
\usepackage{hyperref}
\usepackage{amsfonts}
\usepackage{amssymb}
\usepackage{graphicx}
\usepackage{amsthm}
\usepackage{amsmath}

\newtheorem{definition}{Definition}
\newtheorem{example}{Example}
\newtheorem{theorem}{Theorem}
\newtheorem{lemma}[theorem]{Lemma}
\newtheorem{corollary}[theorem]{Corollary}

\begin{document}
\title{Tight lower bound for percolation threshold on a
  quasi-regular graph}

\author{Kathleen E. Hamilton} 
\affiliation{Department of Physics \&
  Astronomy, University of California, Riverside, California, 92521,
  USA} 
\author{Leonid P. Pryadko} 
\affiliation{Department of Physics \&
  Astronomy, University of California, Riverside, California, 92521,
  USA} 
\date\today
\begin{abstract}
  We construct an exact expression for the site percolation threshold
  $p_c$ on a quasi-regular tree ${\cal T}$, and a related exact lower
  bound for a quasi-regular graph ${\cal G}$.  Both are given by the
  inverse spectral radius of the appropriate Hashimoto matrix used to
  count non-backtracking walks.
  The obtained bound always exceeds the inverse spectral radius of the
  original graph, and it is also generally tighter than the existing
  bound in terms of the maximum degree.
\end{abstract}
\maketitle
An ability to process and store large amounts of information lead to
emergence of big data in many areas of research and applications.
This caused a renewed interest in graph theory as a tool for
describing complex connections in various kinds of networks:
social, biological, technological, 
etc.\cite{Albert-Jeong-Barabasi-2000,%
  Albert-Barabasi-RMP-2002,%
  Borner-Sanyal-Vspignani-ARIST-2007,%
  Danon-etal-2011,Costa-Oliveira-CorreaRocha-2011}
In particular, percolation transition on graphs has been used to
describe internet stability, spread of contagious
diseases, and 
emergence of viral videos.  Percolation has also been
applied to establish the existence of the decoding threshold in
certain classes of quantum error-correcting
codes\cite{Kovalev-Pryadko-FT-2013}.

A \emph{degree} of a vertex in a graph is the number of its neighbors.
Degree distribution is a characteristic easy to extract empirically.
A simple approach for network modeling is to study random graphs with
the given degree distribution\cite{Cohen-Erez-benAvraham-Havlin-2000,%
  Callaway-Newman-Strogatz-Watts-2000,Newman-Strogatz-Watts-2001}.  In
the absence of correlations, the site percolation threshold on such a
random graph
is\cite{Cohen-Erez-benAvraham-Havlin-2000,%
  Callaway-Newman-Strogatz-Watts-2000}
\begin{equation}
  p_c=\frac{\langle d\rangle}{ \langle d^2 \rangle-\langle
d\rangle},\label{eq:rnd-threshold}
\end{equation} 
where $\langle d^m\rangle\equiv \sum_v d_v^m/n$ is the $m$-th moment
of the vertex degree distribution and the \emph{graph order}, $n$, is
the number of vertices in the graph.  While this result is very
appealing in its simplicity, Eq.~(\ref{eq:rnd-threshold}) has no
predictive power for any actual network where correlations between
degrees or enhanced connectivity (``clustering'') of nearby vertices
may be present.  Substantial effort has been spent on attempts to
account for such
correlations\cite{PastorSatorras-Vazquez-Vespignani-2001,%
  Newman-PRL-2009,%
  Nino-MunozCaro-2013} in random graphs.  However, such approaches can
only account for local correlations and are flawed when applied to
artificial networks like the power grid, which may have a carefully
designed robust backbone (e.g., as in Example \ref{ex:one}).  Such
correlations make Eq.~(\ref{eq:rnd-threshold}) or its versions
accounting for local correlations seemingly irrelevant.

There are only a handful of results on percolation for general
graphs\cite{Benjamini-Schramm-1996,Hofstad-2010}.  These include the
exact lower bound for the site percolation threshold for any graph
with the maximum vertex degree 
$d_\mathrm{max}$\cite{Hammersley-1961},%
\begin{equation}
  \label{eq:max-deg-bound}
  p_c\ge (d_\mathrm{max}-1)^{-1},
\end{equation}
which coincides with that for the bond percolation\cite{[{Theorem 1.2
    in }]Hofstad-2010}.  Both bounds are achieved on $d$-regular tree
$\mathcal{T}_d$.  Unfortunately, for graphs with wide degree
distributions, Eq.~(\ref{eq:max-deg-bound}) may easily underestimate
the percolation threshold.

An estimate of the percolation threshold for dense graphs (with some
conditions) as the inverse \emph{spectral radius} of the graph,
$\rho({\cal G})\equiv \rho(A_{\cal G})$, the largest eigenvalue of its
adjacency matrix, $A_{\cal G}$, has been suggested in
Ref.~\onlinecite{Bollobas-Borgs-Chayes-Riordan-2010}.  Unfortunately,
the conditions are rather restrictive, and the
estimate is clearly not very accurate for sparse degree-regular graphs
where the spectral radius $\rho({\cal G})=d$, as this estimate never reaches
the lower bound in Eq.~(\ref{eq:max-deg-bound}).  

\begin{example}
  \label{ex:one}
  Consider a tree graph ${\cal T}\equiv {\cal T}_{d;r,L}$ constructed
  by attaching $r$ chains of length $L$ to each vertex of a
  $d$-regular tree ${\cal T}_{d}$, see
  Fig.~\ref{fig:backbone_network}.  The percolation threshold
  coincides with that of ${\cal T}_d$, $p_c=p_c(T_d)=(d-1)^{-1}$.  On
  the other hand, Eq.~(\ref{eq:rnd-threshold}) gives $p_c \to 0$ if we
  take $L=1$, $r$ large, and $p_c\to 1$ if we take $r=1$, $L$ large.
  Similarly, the spectral radius is $\rho({\cal
    T}_{d;r,1})=d/2+[(d/2)^2+r]^{1/2}$ (we took $L=1$); the
  corresponding estimated threshold varies in the range $0<[\rho({\cal
    G})]^{-1}\le 1/d$, while the lower bound (\ref{eq:max-deg-bound})
  varies in the range $0<p_c^{\rm min}\le (d-1)^{-1}$.
\end{example}
\begin{figure}[htbp]
\centering
\includegraphics[width=1.96in]{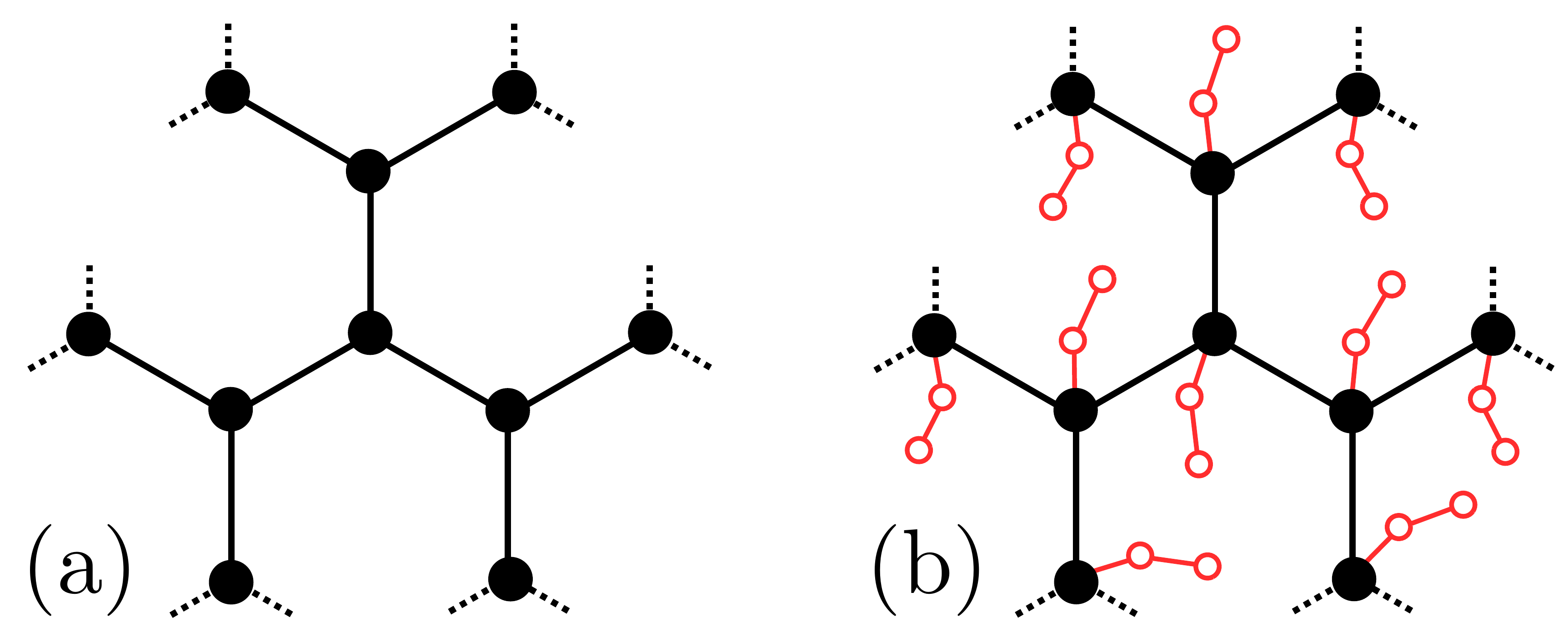}
\caption{(Color online) (a) A $d$-regular tree used for the backbone
  of the graph in Example \ref{ex:one}. (b) The tree
  ${\cal T}_{d;r,L}$ is grown from the backbone by placing $r$ chains
  of fixed length $L$ (shown $d=3$, $r=1$, $L = 2$) at each vertex of
  the backbone.}\label{fig:backbone_network}
\end{figure}
Thus, Eq.~(\ref{eq:rnd-threshold}), the lower bound
(\ref{eq:max-deg-bound}), or the inverse spectral radius
$[\rho({\cal G})]^{-1}$ do not give accurate estimates of the percolation
threshold for this graph family.

In this work we construct an exact expression for the percolation
threshold on any quasi-transitive tree ${\cal T}$, and a related exact
lower bound for the percolation threshold on a quasi-transitive graph
${\cal G}$ which is more specific than Eq.~(\ref{eq:max-deg-bound}).
These are given by the inverse spectral radius of the \emph{oriented
  line graph} (OLG) ${\cal F}$ introduced by Kotani and
Sunada\cite{Kotani-Sunada-2000}.  The corresponding adjacency 
matrix, $A_\mathcal{F}$, is the Hashimoto
matrix\cite{Hashimoto-matrix-1989} used to enumerate non-backtracking
walks on ${\cal G}$.
We also show that
the inverse spectral radius $\rho({\cal G})$ of the original graph
gives a smaller (inexact) lower bound for the percolation threshold,
\begin{equation}
  p_c\ge 1/\rho({{\cal F}})
  > 1/\rho({\cal G}).\label{eq:our-lower-bound}
\end{equation}
We call a graph ${\cal G}=({\cal V},{\cal E})$ with vertex set ${\cal
  V}$ and edge set ${\cal E}$ \emph{transitive} iff for any two
vertices $u$, $v$ in ${\cal V}$ there is an automorphism of ${\cal G}$
mapping $u$ onto $v$.  Graph ${\cal G}$ is \emph{quasi transitive} if
there is a finite set of vertices ${\cal V}_0\subset {\cal V}$ such
that any $v\in {\cal V}$ is taken into ${\cal V}_0$ by some
automorphism of ${\cal G}$.  We say that any vertex which can be
mapped onto a vertex $v_0\in{\cal V}_0$ is in the equivalence class of
$v_0$.  

In site percolation on a graph ${\cal G}$, each vertex is open
with probability $p$ and closed with probability $1-p$; two
neighboring open vertices belong to the same cluster.  Percolation
happens if there is an infinite cluster on ${\cal G}$.

First consider a quasi transitive \emph{tree} ${\cal T}$, a graph with
no cycles.  According to Eq.~(\ref{eq:max-deg-bound}), the
corresponding percolation threshold must be strictly non-zero,
$p_c\equiv p_c({\cal T})>0$.  Percolation threshold on any tree can be
found exactly by constructing a set of recursive equations
starting with some arbitrarily chosen
root\cite{Flory1-1941,*Flory2-1941,*Flory3-1941}.  For a given open
vertex $i$, let us introduce the probability $Q_{ij}$ that $i$ is
connected to a finite cluster through its neighbor $j$.  The
corresponding recursive equations have the form
\begin{equation}
  Q_{ij} = \prod_{l\sim j :  l\neq i}  (1-p + pQ_{jl}),
  \label{eq:MF_eqtn}
\end{equation}
where the product is taken over all neighbors $l$ of $j$ (denoted
$l\sim j$) such that $l\neq i$ so that only so far uncovered
independent branches are included.  The growth of a branch into an
infinite cluster is impeded by a neighboring site being closed
(probability $1-p$), or being open but connecting to a finite branch
(probability $p Q_{jl}$).

Below the percolation threshold, $p< p_c$,
Eqs.~(\ref{eq:MF_eqtn}) are satisfied identically with $Q_{ij}=1$.
Right at the percolation threshold, we expect the probability of an
infinite cluster to be vanishingly small, and the probabilities
$Q_{ij}$ can be expanded
\begin{equation}
Q_{ij} = 1 - \epsilon_{ij},\quad i\sim j,
\end{equation}
where $\epsilon_{ij}$ is infinitesimal.  Expanding
Eqs.~(\ref{eq:MF_eqtn}) to linear order in $\epsilon_{ij}$, we obtain
the following eigenvalue problem 
at the threshold, $p=p_c$,
\begin{equation}
  \label{eq:ev}
  \lambda \epsilon_{ij} = \sum_{l\sim j: l \neq i}  \epsilon_{jl},\quad
  \lambda\equiv 1/p_c.
\end{equation}
The percolation threshold corresponds to the largest eigenvalue
$\lambda$ corresponding to a non-negative eigenvector,
$\epsilon_{ij}\ge 0$.  To ensure the probability $p_c \le 1$, the
eigenvalue needs to be sufficiently large, $\lambda\ge 1$.  It is convenient
to extend Eqs.~(\ref{eq:ev}) to an arbitrary graph ${\cal G}$, where
$\epsilon_{ij}\neq0$ iff the corresponding component of the adjacency
matrix is nonzero, $A_{ij}\neq0$, including any diagonal elements,
$i=j$, corresponding to loops in ${\cal G}$.

The eigenvalue problem (\ref{eq:ev}) has a non-symmetric matrix with
non-negative elements.  According to Perron-Frobenius
theory\cite{Perron-1907,Frobenius-1912,Meyer-book-2000} of
non-negative matrices, there always exists a non-negative solution
with the eigenvalue $\lambda$ equal to the spectral radius $\rho\ge0$ of
this matrix, although in general it is possible to have $\rho=0$.  

To establish a lower bound on $\rho$, we first construct a graphical
interpretation of Eqs.~(\ref{eq:ev}).  The components $\epsilon_{ij}$
correspond to directed edges of the original graph; the entire set
corresponds to sites of the \emph{line digraph}
\cite{Harary-Norman-1960} associated with the symmetric digraph
$\tilde{\mathcal{G}}$ equivalent to the original graph ${\cal G}$.
Namely, each edge $(i,j)\in \mathcal{E}(\mathcal{G})$ is replaced by a
pair of directed edges, $\{(i;j),(j;i)\}\subset
\mathcal{E}(\tilde{\mathcal{G}})$.  The summation over $l$ in the
r.h.s.\ of Eq.~(\ref{eq:ev}) would correspond to the adjacency matrix
of the line digraph of $\tilde{\mathcal{G}}$, were it not for the
exclusion $l\neq i$.  With such a restriction, we obtain the adjacency
matrix of the OLG \cite{Kotani-Sunada-2000} ${\cal F}_{\tilde{\cal
    G}}$ (technically, this is a 
digraph).  

Generally, given a digraph ${\cal D}=({\cal V},{\cal E})$, the
associated OLG ${\cal F}_{\cal D}$ has the vertex set ${\cal V}({\cal
  F}_{\cal D})={\cal E}({\cal D})$ and directed edges $\biglb((i;j);
(j;l)\bigrb)$ such that $\{(i;j),(j;l)\}\subset {\cal E}( {\cal D})$
and $l\neq i$.  To simplify notations, we will call the OLG ${\cal
  F}_{\cal G}$ of a graph ${\cal G}$ the OLG ${\cal F}_{\tilde{\cal
    G}}$ of the corresponding symmetric digraph $\tilde{\cal G}$.  
By construction, Eqs.~(\ref{eq:ev}) are the eigenvalue equations for
the adjacency matrix of the OLG, $A_{{\cal F}_{\cal G}}$.  This matrix
is also known as the Hashimoto matrix of the original
graph\cite{Hashimoto-matrix-1989}.
We will prove the following
\begin{theorem}\label{th:existence}
  The largest real-valued eigenvalue $\lambda$ of Eqs.~(\ref{eq:ev})
  corresponding to a non-trivial eigenvector with non-negative components,
  $\epsilon_{ij}\ge0$, is given by the spectral radius of the OLG,
  $\lambda_\mathrm{max}=\rho({\mathcal{F}_\mathcal{G}})$.  It
  satisfies $\lambda_\mathrm{max}\ge 1$ for any connected quasi
  transitive graph ${\cal G}$ which is not a finite tree.
\end{theorem}

Let $\Gamma$ be a group of automorphisms of a graph ${\cal G}$.  The
\emph{quotient graph} ${\cal G}/\Gamma$ is the graph whose vertices
are equivalence classes ${\cal V}({\cal G})/\Gamma=\{\Gamma
v:v\in{\cal V}({\cal G})\}$, and an edge $(\Gamma u,\Gamma v)$ appears
in ${\cal G}/\Gamma$ if there are representatives $u_0\in \Gamma u$
and $v_0\in \Gamma v$ that are neighbors in ${\cal G}$,
$(u_0,v_0)\in{\cal E}({\cal G})$.  Same definition applies in the case
of a digraph $\mathcal{D}$, except that we need to consider directed
edges, e.g., $(u_0;v_0)\in{\cal E}({\cal D})$.  Notice that a pair of
equivalent neighboring vertices in the original (di)graph ${\cal G}$
produces a loop in the quotient (di)graph ${\cal G}/\Gamma$.  
Notice also that an automorphism $\gamma$ of the digraph
$\mathcal{D}$ 
induces a unique automorphism $\varphi_\gamma$ of the corresponding OLG
$\mathcal{F}_{\cal D}$, and a group $\Gamma$ of automorphisms of
$\mathcal{D}$ induces an isomorphic group $\Phi_\Gamma$ of
automorphisms of $\mathcal{F}_{\mathcal{D}}$.
We will need the following two Lemmas:
\begin{lemma}
  \label{th:backbone-OLG}
  A finite quotient graph $\mathcal{F}_0\equiv \mathcal{F}_{\cal
    G}/\Phi_\Gamma$ of the OLG $\mathcal{F}_{\cal G}$ of any connected
  quasi-transitive graph $\mathcal{G}$ with automorphism group
  $\Gamma$ is strongly connected if the minimum and the maximum vertex
  degrees of ${\cal G}$ satisfy
  $d_\mathrm{max}(\mathcal{G})>2$ and $d_\mathrm{min}(\mathcal{G})>1$.
\end{lemma}
\begin{proof}  
  We are going to prove that for every ordered pair of vertices
  $(u;v)$ in $\mathcal{F}_{\cal G}$ there is a directed path between
  some vertices $u_0\in\Phi_\Gamma v$ and $v_0\in\Phi_\Gamma v$,
  respectively equivalent to $u$, $v$ under the isomorphism group
  $\Phi_\Gamma$ induced by $\Gamma$.  The vertices $\{u, v\} \subset
  \mathcal{F}_{\cal G}$ are directed edges in the digraph $\tilde
  {\cal G}$; denote the corresponding undirected edges
  $\{e_u,e_v\}\subset{\cal E}({\cal G})$.  Connectivity of
  $\mathcal{G}$ implies the existence of a path on ${\cal G}$
  connecting a vertex in $e_u$ and a vertex in $e_v$ which does not
  include these two edges.  Thus, there is a directed path on
  $\tilde{\cal G}$ connecting either $u$ or reverse of $u$ with either
  $v$ or reverse of $v$.  

  To ensure the existence of a directed path
  between the directed edges equivalent to actual $u$ and $v$, we may
  just construct a directed non-backtracking path from any directed
  edge $u\in{\cal E}(\tilde{\cal G})$ to one of the edges $\bar u_0\in
  \Phi_\Gamma \bar u$ equivalent to its reverse, $\bar u$.  Since there are no
  degree-one vertices on ${\cal G}$, with a finite number of vertex
  equivalence classes induced by $\Gamma$, any non-backtracking path
  $p$ starting with $u$ will eventually come to a vertex equivalent to
  that already in the path; the corresponding path $\Gamma p$ on the
  quotient graph $\mathcal{G}/\Gamma$ loops back onto itself.
  Consider two such paths $p_1$ and $p_2\neq p_1$ starting at $u$;
  they exist since there is at least one vertex with degree $d>2$ in
  the graph $\mathcal{G}$.  If either of $\Gamma p_1$ or $\Gamma p_2$
  loops back onto itself on $\mathcal{G}/\Gamma$ at a point other than
  the tail of $u$, we can complete a portion of that path in the
  reverse direction to arrive at some $\bar u_0$ equivalent to the reverse
  of $u$.  Otherwise (both $p_1$ and $p_2$ end at equivalents of the
  tail of $u$), the required path on $\mathcal{G}/\Gamma$ is $\Gamma
  p_1$ joined with the reverse of $\Gamma p_2$, with any backtracking
  segments in the resulting path removed.  In either case, the
  corresponding path on $\mathcal{F}_\mathcal{G}$ connects $u$ with an
  equivalent of its reverse, $\bar u$; its image under $\Phi_\Gamma$ is
  the path connecting $u\in \mathcal{F}_\mathcal{G}/\Phi_\Gamma$ with
  the corresponding reverse, $\bar u\in
  \mathcal{F}_\mathcal{G}/\Phi_\Gamma$.
\end{proof}
\begin{lemma}
  \label{th:backbone-ev}
  A non-trivial solution of Eqs.~(\ref{eq:ev}) with an eigenvalue
  $\lambda\ge 1$ and a positive-component eigenvector 
  satisfying the condition $\epsilon_{ij}=\epsilon_{i'j'}$ for any two
  ordered pairs of adjacent vertices $(i;j)$ and $(i';j')$ that can be
  mapped onto each other by some automorphism of $\mathcal{G}$ exists
  and is unique for any connected quasi-transitive graph $\mathcal{G}$
  with vertex degrees limited by
  $d_\mathrm{max}(\mathcal{G})>2$ and $d_\mathrm{min}(\mathcal{G})>1$.
\end{lemma}

\begin{proof} The ansatz leaves a finite eigensystem with a matrix $M$
  whose non-zero elements correspond to the adjacency matrix of the
  quotient graph of the OLG, $\mathcal{F}_0\equiv
  \mathcal{F}_\mathcal{G}/\Phi_\Gamma$.  The statement of the Lemma
  follows from Lemma \ref{th:backbone-OLG} and
  Perron-Frobenius
  theorem\cite{Perron-1907,Frobenius-1912,Meyer-book-2000}.  
\end{proof}
Note that the eigenvalue in Lemma \ref{th:backbone-ev} is
given by the spectral radius of $M$ and is bounded from above and
below by the spectral radii of $\mathcal{F}_\mathcal{G}$ and
$\mathcal{F}_0$, respectively:
\begin{equation}
  1\le   \rho({  \mathcal{F}_\mathcal{G}/\Phi_\Gamma})
  \le \lambda=\rho(M)
  \le \lambda_\mathrm{max}=\rho({\mathcal{F}_\mathcal{G}}). 
  \label{eq:spectral-ineq-long}
\end{equation}
\begin{proof}[Proof of Theorem \ref{th:existence}]
  Define a {\em backbone\/} $\mathcal{B}$ of a graph $\mathcal{G}$, a
  result of the recursive removal of all degree-one vertices.
  For any finite tree the backbone is empty.  For a connected graph
  $\mathcal{G}$ which is not a finite tree, the backbone $\mathcal{B}$
  satisfies $d_\mathrm{min}(\mathcal{B})>1$.  If $\mathcal{G}$ is a
  connected quasi-transitive graph, so is $\mathcal{B}$.  If, in
  addition, $d_\mathrm{max}(\mathcal{B})>2$, then the backbone
  $\mathcal{B}$ satisfies the conditions of Lemma \ref{th:backbone-ev}
  which gives an explicit solution in this case.  Otherwise,
  $\mathcal{B}$ is a connected degree-regular graph with
  $d_\mathrm{max}=d_\mathrm{min}=2$; it is a simple cycle or an
  infinite chain.  In this case the adjacency matrix
  $A(\mathcal{F}_\mathcal{B}/\Phi_\Gamma)$ has two independent
  strongly-connected components corresponding to the two classes of
  non-backtracking paths on $\mathcal{B}$; the corresponding
  eigenvalue $\rho(\mathcal{F}_\mathcal{B}/\Phi_\Gamma)=1$ is
  doubly-degenerate.  In either case, the only admissible eigenvalue
  is given by the spectral radius, $\lambda_\mathrm{max}\ge\lambda=
  \rho(M)\ge\rho({\mathcal{F}_\mathcal{B}/\Phi_\Gamma})\ge1$, see
  Eq.~(\ref{eq:spectral-ineq-long}).
  
  The original graph $\mathcal{G}$ can be restored from the backbone
  $\mathcal{B}$ by restoring degree-one vertices in the opposite order
  starting from the last removed.  For such a vertex $v$ which is
  connected to the vertex $u$ already in the graph, we notice that
  $\epsilon_{uv}=0$ (this bond cannot lead to an infinite cluster), while
  $\epsilon_{vu}$ is determined by the values $\epsilon_{uj}$ for bonds
  already in the graph, see Eq.~(\ref{eq:ev}), where the same eigenvalue
  $\lambda$ must be used.  Thus, additional vertices in
  $\mathcal{F}_\mathcal{G}$ cannot modify the components
  $\epsilon_{ij}\ge 0$ with $\{i,j\}\subset\mathcal{V}(\mathcal{B})$,
  and the maximum eigenvalue remains the same,
  $\lambda_\mathrm{max}=\rho({\mathcal{F}_\mathcal{G}})
  =\rho({\mathcal{F}_\mathcal{B}})\ge
  \rho({\mathcal{F}_\mathcal{B}/\Phi_\Gamma})\ge1$.
\end{proof}
We next apply the constructed mean field theory to calculating
percolation thresholds of more general graphs which may contain cycles.
The main result of this work will be the following 
\begin{theorem}
  The percolation threshold for any simple quasi-transitive graph
  $\mathcal{G}$ which is not a finite tree is bounded from below by
  the inverse spectral radius of the corresponding OLG, $p_c(\mathcal{G})\ge
  1/\rho(\mathcal{F}_\mathcal{G})$. 
  \label{th:threshold-bound}
\end{theorem}
This bound is tight, as it becomes an equality for quasi-transitive
trees.  The approach is to construct a tree graph $\mathcal{T}$ which
is locally indistinguishable from the original graph $\mathcal{G}$,
except that a closed walk on $\mathcal{G}$ goes over to a walk connecting
equivalent points on $\mathcal{T}$.  We start by defining an operation
for single cycle unwrapping (SCU) at a given bond which is not a
bridge:
 
\begin{definition}
  Given a connected graph $\mathcal{G}$ and a bond $b\equiv
  (u,v)\in\mathcal{E}(\mathcal{G})$, such that the \emph{two-terminal}
  graph $\mathcal{G}'\equiv
  \biglb(\mathcal{V}(\mathcal{G}),\mathcal{E}(\mathcal{G})\setminus
  b\bigrb)$ with source at $v$ and sink at $u$ is connected, define
  the cycle-unwrapped graph $\mathcal{C}_b\mathcal{G}$ as the series
  composition of an infinite chain of copies $\mathcal{G}'_i$,
  $i\in\mathbb{Z}$, of the graph $\mathcal{G}'$, with the source of
  $\mathcal{G}'_i$ connected to the sink of the $\mathcal{G}'_{i+1}$.
\end{definition}  
The SCU is illustrated in Fig.~\ref{fig:loop_unwrap_1}.  Notice that
for a graph with more than one cycle, unwrapping at $b$ removes one
cycle but creates an infinite number of copies of the remaining
cycles.  Nevertheless, for a locally finite graph, a countable number
of SCUs is needed to remove all cycles.  Indeed, the cycle-unwrapped
image of any path on $\mathcal{G}$ that does not include $b$ will
remain entirely within a single copy of $\mathcal{G}'$.  Thus, if at
each SCU step we choose a bond $b$ at distance $r_b$ from some fixed
origin vertex, such that only bridge bonds can be found closer to the
origin, $r<r_b$, any copy of the remaining non-bridge bond introduced
by the SCU is going to be at a distance $r>r_b$.  Thus, each SCU reduces
the number of non-bridge bonds at $r_b$, and for a locally finite
graph, a finite number of SCUs is required to ensure that all bonds at
$r=0, 1, 2, \ldots$ are bridge bonds.  This proves
\begin{lemma}
  \label{th:sequence}
  For a locally finite graph $\mathcal{G}_0\equiv \mathcal{G}$, a
  sequence of SCUs
  $\mathcal{G}_{m+1}\equiv\mathcal{C}_{{b}_{m+1}}\mathcal{G}_m$ can
  be chosen so that in the $m\to\infty$ limit the resulting graph is a
  tree, $\mathcal{T}\equiv \mathcal{C}_\infty\mathcal{G}$.
\end{lemma}

\begin{figure}[htbp]
\centering
\includegraphics[width=\columnwidth]{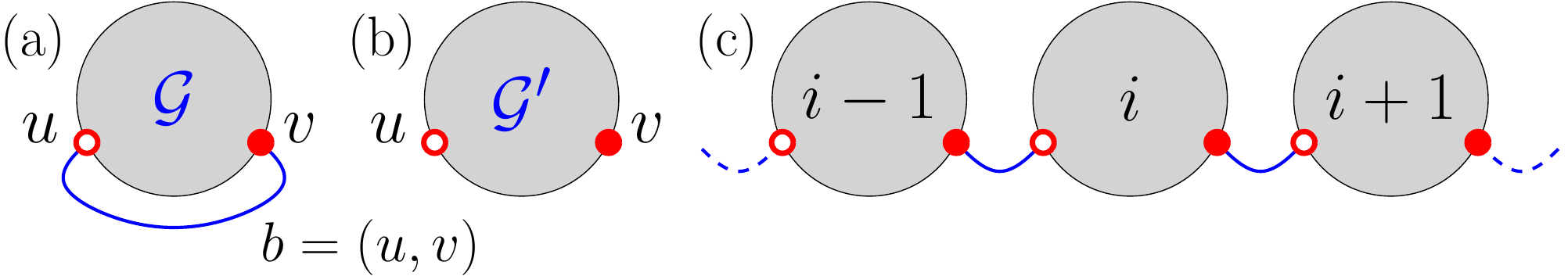}
\caption{(Color online) Illustration of SCU: (a) A graph $\mathcal{G}$ with a
  non-bridge bond $b\equiv (u,v)$ highlighted; (b) Two-terminal graph
  $\mathcal{G}'$; (c)~The resulting graph $\mathcal{C}_b\mathcal{G}$ 
  is a series composition of an infinite chain of copies of
  $\mathcal{G}'$.}
\label{fig:loop_unwrap_1}
\end{figure}

Clearly, the graph $\mathcal{C}_b\mathcal{G}$ produced by an SCU
has a group of isomorphisms $\mathbb{Z}$ generated by the translation
$i\to i+1$.  The corresponding graph quotient recovers the original
graph, $\mathcal{G}=(\mathcal{C}_b\mathcal{G})/\mathbb{Z}$.  This
symmetry allows us to prove the following 
\begin{lemma}
  For any simple quasi-transitive graph $\mathcal{G}$, SCU does not
  change the spectral radius of $\mathcal{G}$,
  $\rho(\mathcal{G})=\rho(\mathcal{C}_b\mathcal{G})$, or of the
  OLG,
  $\rho(\mathcal{F}_\mathcal{G})
  =\rho(\mathcal{F}_{\mathcal{C}_b\mathcal{G}})$. 
  \label{th:spectral-radius}
\end{lemma}
\begin{proof}
  The symmetry of $\mathcal{C}_b\mathcal{G}$ implies that an
  eigenvector $e$ can always be chosen to diagonalize both its
  adjacency matrix $A\equiv A(\mathcal{C}_b\mathcal{G})$, $A e=\lambda
  e$, and the translation generator $T$, $T e=\mu e$.  Translation
  group being Abelian, its representations are all one-dimensional,
  with $\mu=e^{ik}$, with $0\le k<2\pi$.  Let $e_0$ with non-negative
  components be the
  Perron-Frobenius eigenvector of the 
  non-negative matrix $A$ with the eigenvalue equal to its spectral
  radius, $\lambda_\mathrm{max}=\rho(A)$.  Symmetrizing $e_0$ over
  $\mathbb{Z}$, gives a non-negative eigenvector $e$
  corresponding to the same $\lambda_\mathrm{max}$ and $k=0$.  This
  corresponds to the ansatz introduced in Lemma \ref{th:backbone-ev},
  thus $\rho(M')=\rho({\mathcal{C}_b\mathcal{G}})$, where $M'$ is the
  reduced matrix corresponding to the symmetric eigenvector of $A$,
  cf.\ Eq.~(\ref{eq:spectral-ineq-long}).  Further, for a simple (di)graph
  $\mathcal{G}$, the matrix elements of $M'$ satisfy
  $M'_{ij}\in\{0,1\}$, thus $M'=A_\mathcal{G}$, which gives
  $\rho(\mathcal{G})=\rho(M')=\rho(\mathcal{C}_b\mathcal{G})$.  The
  proof in the case of $\mathcal{F}_{\mathcal{C}_b\mathcal{G}}$ is
  identical if we notice
  $\mathcal{C}_b\mathcal{F}_\mathcal{G}=\mathcal{F}_{\mathcal{C}_b\mathcal{G}}$.
\end{proof}
\begin{proof}[Proof of Theorem \ref{th:threshold-bound}] Vertex
  transitivity of $\mathcal{G}$ implies that a finite maximum degree
  exisits; according to Lemma~\ref{th:sequence} $\mathcal{G}$ can be
  transformed to a tree $\mathcal{T}$ by a series of SCUs.  Each step
  of the sequence can be undone by a graph quotient,
  $\mathcal{G}_{m}=\mathcal{G}_{m+1}/\mathbb{Z}$.  According to
  Theorem 1 in Ref.~\onlinecite{Benjamini-Schramm-1996}, the percolation
  threshold of a graph quotient cannot be below that of the original
  graph, thus
  $p_c(\mathcal{G}_{m})=p_c(\mathcal{G}_{m+1}/\mathbb{Z})\ge
  p_c(\mathcal{G}_{m+1})$; the entire sequence gives
  $p_c(\mathcal{G})\ge p_c(\mathcal{T})$.  On the other hand,
  Eq.~(\ref{eq:ev}) and Theorem \ref{th:existence} give
  $p_c(\mathcal{T})=1/\rho(\mathcal{F}_\mathcal{T})$.  Moreover, each
  of the intermediate graphs of the sequence is vertex transitive and
  simple, thus the spectral radius of corresponding OLG is preserved
  at each step,
  $\rho(\mathcal{F}_\mathcal{T})=\rho(\mathcal{F}_\mathcal{G})$, see
  Lemma \ref{th:spectral-radius}.
\end{proof}

Finally, we establish the relation between the spectral radius of OLG
with that of the original graph: 
\begin{theorem}
  The spectral radius of any connected non-empty graph $\mathcal{G}$ is
  strictly larger than that of the corresponding OLG,
  $\rho(\mathcal{G})> \rho(\mathcal{F}_\mathcal{G})$.
\end{theorem}
\begin{proof}
  A non-empty graph contains at least one edge (or a loop), thus
  $\rho({\cal G})>0$; we only need to consider the case where
  $\rho({\cal F}_{\cal G})>0$.  Begin with Eq.~(\ref{eq:ev}) and
  assume $\epsilon_{ij}\ge0$ is the non-zero eigenvector corresponding
  to the eigenvalue $\lambda\equiv \rho({\cal F}_{\cal G})>0$.
  Introduce vertex variables
  \begin{equation}
    \label{eq:vertex-vars}
    y_i\equiv  \sum_{j:j\sim i}\epsilon_{ij},
  \end{equation}
  corresponding to the sum of $\epsilon_{ij}$ over all directed bonds
  leaving a given vertex $i$.  These variables
  satisfy\cite{Krzakalaa-etal-2013} 
\begin{equation}
  \label{eq:ev-A}
  \left[\lambda^2 I+{(D-I)}\right]y=\lambda\, Ay,
\end{equation}
where $D\equiv \diag (d_1,\ldots,d_n)$ is the diagonal matrix of
degrees, $I$ is the identity matrix, and $A\equiv A_{\cal G}$ is the
(symmetric) adjacency matrix of ${\cal G}$.  If we multiply
Eq.~(\ref{eq:ev-A}) by $y^T$ on the left, the r.h.s.\ does not exceed
$\lambda\rho(A)\|y\|^2$.  From the proof of Theorem \ref{th:existence}
it follows that there exists a vertex on the backbone of ${\cal G}$
with $d_v>1$ such that the corresponding component of $y$ is non-zero,
thus $y^T(D-I)y>0$; dropping this term gives $\lambda<\rho(A)$.
\end{proof}
\begin{corollary}
  The percolation threshold for any quasi-transitive graph that is not a
  finite tree satisfies Eq.~(\ref{eq:our-lower-bound}).
\end{corollary}

In conclusion, we constructed an exact expression for the threshold of
site percolation on an arbitrary quasi-transitive tree, and an
associated exact lower bound for such a threshold on an arbitrary
graph.  These are given by the inverse spectral radius of the oriented
line graph associated with the tree or the graph, respectively.  The
constructed bound accounts for local structure of the graph, and is
asymptotically exact for graphs with no short loops.  For
degree-regular graphs it goes over into the known lower bound
(\ref{eq:max-deg-bound}).  We also demonstrated that the inverse
spectral radius of the original graph ${\cal G}$ which was suggested
previously as an estimate for the percolation threshold is always
strictly smaller than our lower bound, see
Eq.~(\ref{eq:our-lower-bound}).  In applications, spectral radius for
sparse graphs involving billions of edges can be readily evaluated
using standard numerical packages.

Our results can be easily extended to the cases of Bernoulli (bond),
combined site-bond, or non-uniform percolation, where the
probabilities to have an open vertex may differ from site to site.  A
similar technique can also be used to prove the conjecture on the
location of the threshold for vertex-dependent percolation on directed
graphs\cite{Restrepo-Ott-Hunt-2008}.

{\bf Acknowledgments.}  This work was supported in part by the
U.S. Army Research Office under Grant No.\ W911NF-11-1-0027, and by
the NSF under Grant No.\ 1018935.  LP also acknowledges hospitality by
the Institute for Quantum Information and Matter, an NSF Physics
Frontiers Center with support of the Gordon and Betty Moore
Foundation.

%\bibliography{lpp,percol,qc_all,more_qc,percol_append} 

\end{document}